%% file: main.tex
\documentclass[10pt,conference,letterpaper]{IEEEtran}

\usepackage{cite}
\usepackage[pdftex]{graphicx}
\usepackage[cmex10]{amsmath}
\usepackage[boxed]{algorithm}
\usepackage{algorithmic}
\usepackage{array}
\usepackage{mdwmath}
\usepackage{mdwtab}
\usepackage{amssymb}
\usepackage{amsthm}
\usepackage{bbm}

\usepackage[tight,footnotesize]{subfigure}
\usepackage[caption=false,font=footnotesize]{subfig}

\usepackage{url}

\newtheorem{definition}{Definition}
\newtheorem{theorem}{Theorem}
\newtheorem{assumption}{Assumption}
\newtheorem{lemma}{Lemma}

\input{SuchaLib.tex}

\begin{document}

\title{Throughput-Optimal Load Balancing\\for Intra Datacenter Networks}

\author{
  \IEEEauthorblockN{Sucha Supittayapornpong,~~Michael J. Neely}
  \IEEEauthorblockA{University of Southern California\\
    Email: supittay@usc.edu,~~mjneely@usc.edu}
}

\maketitle

\begin{abstract}
Traffic load-balancing in datacenters alleviates hot spots and improves network utilization.  In this paper, a stable in-network load-balancing algorithm is developed in the setting of software-defined networking.  A control plane configures a data plane over successive intervals of time.  While the MaxWeight algorithm can be applied in this setting and offers certain throughput optimality properties, its \emph{bang-bang} control structure rewards single flows on each interval and prohibits link-capacity sharing.  This paper develops a new algorithm that is throughput-optimal and allows  link-capacity sharing, leading to low queue occupancy.  The algorithm deliberately imitates weighted fair queueing, which provides fairness and graceful interaction with TCP traffic.  Inspired by insights from the analysis, a heuristic improvement is also developed to operate with practical switches and TCP flows.  Simulations from a network simulator shows that the algorithm outperforms the widely-used equal-cost multipath (ECMP) technique.
\end{abstract}

\section{Introduction}
Datacenter networks serve as infrastructure for search engines, social networks, cloud computing, etc.  Due to potentially 
high traffic loads, load-balancing becomes an important solution to improve network utilization and alleviate hot spots \cite{VL2,ScalableTopology,GoogleDC,FBDC}.  A widely-used technique is equal-cost multipath (ECMP), where traffic flows are split equally according to the number of available equal-cost next-hops.  However, ECMP does not take into account actual traffic and is susceptible to asymmetric topology \cite{Niagara,CONGA,WCMP}.  Further, the deployment of ECMP is limited due to its equal-cost constraint \cite{BCube}.

Traffic load-balancing can be implemented using software-defined networking (SDN).  An SDN switch, a network device supporting layer-2 and layer-3 operations in OSI architecture, consists of a data plane and control plane \cite{SDNSwitch}.\footnote{An SDN switch should not be confused with a crossbar switch, which may reside in the data plane of an SDN switch.}   The data plane forwards packets according to given rules and operates at a fast timescale, e.g., $1$ns.  The control plane sets those rules and operates at a much slower timescale, e.g., $1$ms.  Several traffic load-balancing algorithms can be implemented through the control plane reconfiguration.

Existing traffic load-balancing methods for datacenter networks distribute traffic according to network capacity and measured traffic.  Weighted-cost multipath \cite{WCMP} distributes traffic according to path capacity.  Centralized algorithms, such as Hedera \cite{Hedera} and Niagara \cite{Niagara}, take advantage of global traffic information to split traffic at a coarse timescale.  Load-balancing with a finer timescale has been implemented in DeTail \cite{DeTail} and CONGA \cite{CONGA} using the \emph{in-network} technique, where decisions are made at switches inside a network without any central controller.  Conceptually, these approaches attempt to distribute traffic over available network resources.  However, the path-based approach in CONGA limits scalability, and the packet-by-packet dispersion in DeTail needs TCP with out-of-order resilience.  Further, these algorithms do not come with analytical optimality proofs, and it is not clear if they are \emph{throughput optimal}.

An algorithm is throughput optimal if it stably supports any feasible traffic load, so that average backlog is bounded \cite{MaxWeight}.  Specifically, a throughput-optimal algorithm utilizes the entire network capacity and can distribute traffic to any portion of the network to maintain network stability.  MaxWeight \cite{MaxWeight} is a well-known throughput-optimal algorithm and has been studied for packet radio \cite{MaxWeight}, switching \cite{McKeown:Switch}, and inter-datacenter networking \cite{Javidi:DCNet}.  It has been generalized to optimize power allocation \cite{Neely:Power}, throughput \cite{Eryilmaz:Throughput}, etc.  Practical aspects of MaxWeight such as finite buffer capacity and fairness with TCP connections have been studied in \cite{Sucha:FIFO,WirelessFiniteBuff,TCP-Backpressure}.  However, MaxWeight is not suitable for   
in-network load-balancing because it prohibits the sharing of link capacity at the data plane's timescale and thus causes high queue occupancy.  This queue size has a finite average, but the size of the longer timescale makes that average unacceptably large.

\begin{figure}[!t]
  \centering
  \includegraphics[scale=1.0]{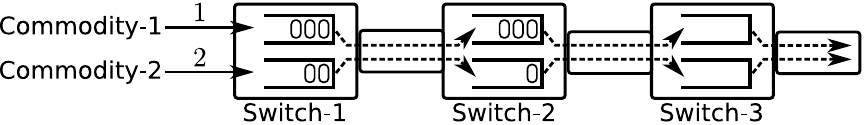}\vspace{-2mm}
  \caption{MaxWeight example: Let $w_{ij}^d$ be the weight of commodity $d$ over the link from switch $i$ to switch $j$.  All weights in this figure are $\prtr{ w_{12}^1, w_{12}^2, w_{23}^1, w_{23}^2 } = \prtr{ 0, 1, 3, 1 }$.}\vspace{-4.5mm}
  \label{fig:ToyExample}
\end{figure}

\begin{figure*}[!t]
  \centering
  \includegraphics[scale=1.0]{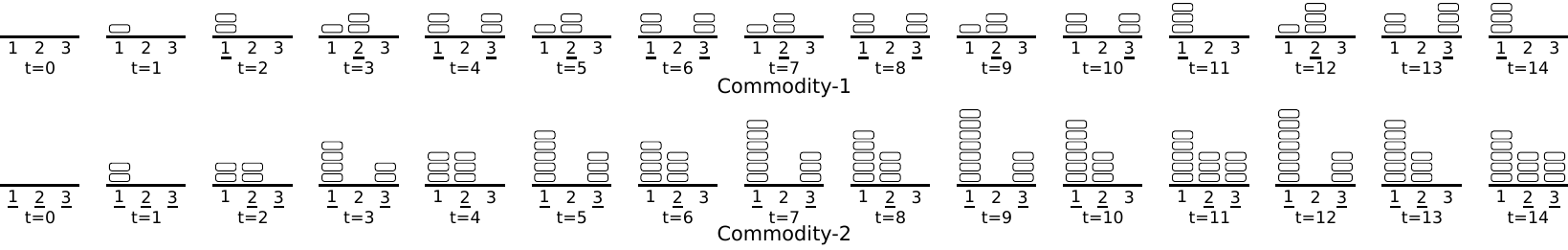}\vspace{-2mm}
  \caption{Timeline of queue occupancy under MaxWeight: a small box represents a packet in a queue.  Switch $i$ is represented by the number $i \in \prtc{1,2,3}$ under the long line.  The short line under the number indicates that the commodity at the numbered switch is served in that particular time slot.  The occupancy pattern repeats after $t=14$, which is similar to the pattern at $t=11$.}\vspace{-4.5mm}
  \label{fig:MaxWeightTimeLine}
\end{figure*}

The MaxWeight algorithm is illustrated by the example in Fig. \ref{fig:ToyExample}.  Two traffic commodities share  three links passing through switches $1,2$ and $3$.  Time is slotted.  The slot size equals the length of the decision update 
interval (the control plane's timescale).  The capacity of each link is $3$ packets per slot.  Every switch has a dedicated queue for each commodity.  In every time slot, MaxWeight calculates, for each link and commodity, a weight equal to the differential backlog between a queue and its next-hop queue.  For that slot, the entire capacity of the link is allocated to the commodity with the maximum non-negative weight, while a commodity with negative weight is ignored.  For example,
 commodity $2$ is served on the link between switches $1$ and $2$ in Fig. \ref{fig:ToyExample}, and 
 commodity $1$ is served on the link between switches  $2$ and $3$.

Let the arrival rates to switch $1$ of commodities $1$ and $2$ be respectively $1$ and $2$ packets per slot.  The timeline of queue evolution is shown in Fig. \ref{fig:MaxWeightTimeLine}.  MaxWeight is effective and always transmits three packets per slot after $t=11$.  This effectiveness requires a sufficient amount of queue backlog. For example, commodity $2$ backlog at switch $1$ is always at least $5$ for times $t  \geq 11$.  This occupancy might be acceptable.  However, if the control plane is reconfigured at a slow timescale relative to the link 
capacity, the queue occupancy can be very high. For example, with a $1$ms update interval, $10$Gbps link speed, and $1$kB packet size, each link can serve 1250 packets per slot (rather than just 3).  This multiplies queue backlog in the 
timeline of Fig. \ref{fig:MaxWeightTimeLine} by a factor $1250/3$,  so the minimum queue backlog of commodity $2$ at switch $1$ is $5\times(1250/3)\approx 2083$ for $t \geq 11$.  Another undesirable property of MaxWeight is that 
queue occupancy  scales linearly with the number hops, as shown in \cite{Scott:LIFO,ShadowQueue}.  In fact, Fig. \ref{fig:MaxWeightTimeLine} is inspired by an example in \cite{Scott:LIFO}.

In practice, the MaxWeight mechanism with a long update interval leads to i) large buffer memory, ii) packet drops, iii) high latency, and iv) burstiness (no capacity sharing during an update interval).  Issues (i)--(iii) can be alleviated partially by the techniques in \cite{Neely:SNO,Sucha:Quadratic,Longbo:Lagrange,ShadowQueue}.  However, issue (iv) resides in the decision making mechanism of MaxWeight and still persists under those techniques.  The situation is worse when issues (ii) and (iv) interact with TCP congestion control, causing slow flow rate and under utilization.  To put it into theoretical perspective, even though MaxWeight solves a network stability problem with $O(1)$ average queue size, the constant factor is too large for a practical system with a long update interval.  Note that an ideal algorithm for the example in Fig. \ref{fig:ToyExample} always serves $1$ and $2$ packets of commodities $1$ and $2$ by sharing the link capacity as shown in Fig. \ref{fig:IdealAlgorithm}.

\begin{figure}
  \centering
  \includegraphics[scale=1.0]{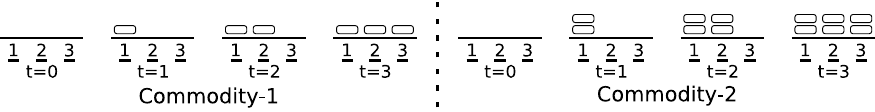}\vspace{-2mm}
  \caption{Timeline of queue occupancy under the ideal algorithm}\vspace{-4.5mm}
  \label{fig:IdealAlgorithm}
\end{figure}

In this paper, a new throughput-optimal algorithm is developed.  The algorithm shares link capacity among commodities during an update interval, resulting in low queue occupancy and low latency.  The key challenge is to design a model and an algorithm that are analyzable, provably optimal,  and practically implementable at the same time.  The algorithm imitates the weighted fair queueing (WFQ) \cite{WFQ,GPS} available in practical switches to provide fairness and low latency among TCP flows in practice.  A general intra data center network may have an exponential number of paths, and our algorithm comes with an optimality proof considering all possible paths using per-commodity queueing which grows linearly with the number of commodities.  This is also a key distinct aspect from the path-based algorithm in \cite{PathBase}.

Section \ref{sec:IdealAlgorithm} develops the throughput-optimal algorithm. 
Inspired by this algorithm, Section \ref{sec:Implementation} presents an enhanced
algorithm that includes heuristics  to cope with practical aspects, including queue information dissemination, queue approximation, and packet reordering issues in TCP.  This heuristic algorithm uses local queue information and local measured traffic to hash TCP flows to next-hop switches and set weights of the weighted fair queueing.  The hash-based mechanism is chosen to reduce packet reordering, which is not possible for DeTail \cite{DeTail}.  Simulation results in Section \ref{sec:Simulations} 
show that both proposed algorithms outperform MaxWeight and ECMP algorithms in ideal simulation and in more realistic simulation with OMNeT++\cite{OMNet}.

\section{System Model and Design}
\label{sec:Model}
The fast timescale of the data plane operates over slotted time $t \in \IntSet_+$, where $\IntSet_+ = \prtc{0, 1, 2, \dotsc }$.  The control plane configures the data plane every $T$ slots, where $T$ is a positive integer.  Thus, reconfigurations occur at times in the set $\set{T} = \prtc{0, T, 2T, \dotsc}$.

\subsection{Topology and Routing}
An intra datacenter network is the interconnection of switches and destinations (such as servers), as shown in Fig. \ref{fig:ModelTopology}.
Traffic designated for a particular destination $d$ is called \emph{commodity $d$} traffic.  Let $\set{S}$ be the set of all switches and $\set{D}$ be the set of all destinations (commodities).
A link between switches $i$ and $j$ is bi-directional with capacity $c_{ij}$ from $i$ to $j$, and capacity $c_{ji}$ in the reverse direction. 
Define $c_{ij} = 0$ if  link $(i,j)$ does not exist or if $i=j$.

\begin{figure}
  \centering
  \includegraphics[scale=1.0]{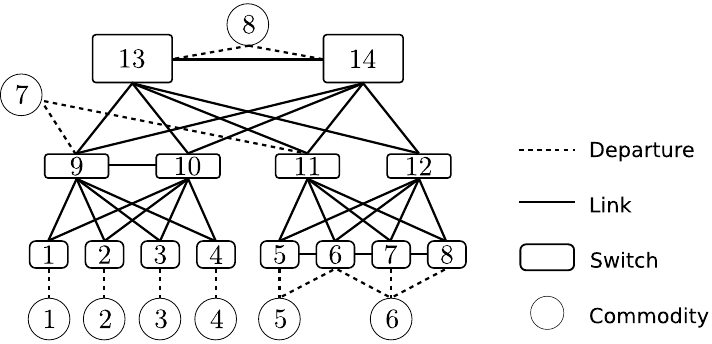}\vspace{-2mm}
  \caption{An example network with $\set{S} = \prtc{1, 2, \dotsc, 14}$ and $\set{D} = \prtc{1, 2, \dotsc, 8}$}\vspace{-4.0mm}
  \label{fig:ModelTopology}
\end{figure}

Each switch must decide where to send its packets next.  Define 
$\set{H}_i^d \subseteq \set{S}$ as the set of next-hop switches available to commodity $d$ packets in switch $i$, 
for all $i \in \set{S}$ and $d \in \set{D}$.  
In practice, these sets can be obtained from manual configuration or from other routing mechanisms.  Define the set of all next-hop switches from switch $i$ as $\set{H}_i = \cup_{d \in \set{D}} \set{H}_i^d$, and the set of previous-hop switches as $\set{P}_i^d = \prtc{ j \in \set{S} : i \in \set{H}_j^d }$ for $i \in \set{S}, d \in \set{D}$.  These sets are illustrated in Fig. \ref{fig:Nexthop}.   Define $\set{D}_{ij} = \prtc{ d \in \set{D} : j \in \set{H}_i^d}$ as the set of all commodities utilizing a link from switch $i$ to switch $j$ for $i,j \in \set{S}$.  Note that this model allows arbitrary path lengths, and can be applied to existing topologies in \cite{VL2,ScalableTopology,GoogleDC,FBDC,BCube}.

\begin{figure}
  \centering
  \includegraphics[scale=1.0]{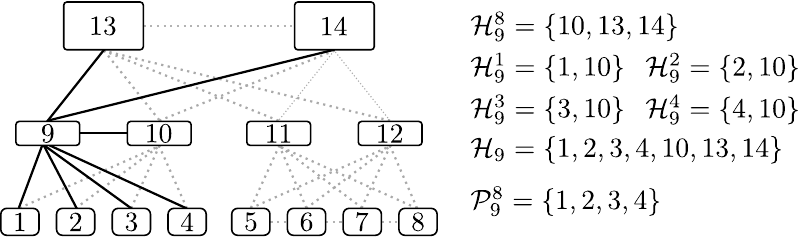}\vspace{-2mm}
  \caption{Example of sets of switches at switch 9.  Note that $\set{P}_9^8$ must not contain $10$ to avoid loops and imposes that $9 \notin \set{H}_{10}^8$.}\vspace{-4.5mm}
  \label{fig:Nexthop}
\end{figure}

\subsection{Traffic}
Switch $i$ receives $a_i^d(t)$ commodity-$d$ packets from external sources at time $t$ (for $i \in \set{S}, d \in \set{D}$).  The external source represents a group of servers or a link connecting to the outside of the network.  For each $i \in \set{S}$, $d \in \set{D}$, the  arrival process $\{a_i^d(t)\}_{t=0}^{\infty}$ is independent and identically distributed (i.i.d.) across time slots.  The i.i.d. assumption is useful for a simple and insightful analysis. 
The resulting algorithm developed under this assumption inspires a heuristic algorithm in Section \ref{sec:Implementation} that does not require i.i.d. arrivals and works gracefully with TCP traffic. 

Recall that $c_{ij}$ is the capacity of the link between switch $i$ and switch $j$, for $i, j \in \set{S}$. 
Let $b_i^d$ be the capacity of the link between switch $i \in \set{S}$ and destination $d \in \set{D}$. Set $b_i^d=0$ if
switch $i$ does not have a direct link
to destination $d$.   Assume arrivals and link capacities are  always bounded by a constant $\delta>0$, so  $0\leq c_{ij} \leq \delta$, 
$0\leq a_i^d(t) \leq \delta$,  $0\leq b_i^d \leq \delta$ for all $i,j\in \set{S}, d \in \set{D}, t \in \IntSet_+$.

\subsection{Decision Variables}

Decision variables are defined for every link connecting switch $i$ to its next-hop switch $j$, for $i \in \set{S}, j \in \set{H}_i$.  Recall that $\set{D}_{ij}$ is a set of commodities using the link.  At configuration time $t \in \set{T}$, the control plane in switch $i$ chooses a \emph{control plane decision variable} $x_{ij}^d(t,T)$ for $d \in \set{D}_{ij}$, which represents a constant transmission rate allocated to commodity $d$ (in units of packets) for the entire $T$-slot interval.  Define $x_{ij}^d(t,T) = 0$ for $d \in \set{D} \backslash \set{D}_{ij}$.  The control plane decisions for  link $(i,j)$ are chosen to satisfy the link capacity constraint: 
$$ \sum_{d \in \mathcal{D}_{ij}}  x_{ij}^d(t,T) \leq T c_{ij}. $$

Once $x_{ij}^d(t,T)$ is determined, no more than $x_{ij}^d(t,T)$  commodity-$d$ packets can be  
transmitted by the data plane during an interval $\prtc{t, \dotsc, t+T-1}$.  For example, the data plane can impose a token bucket mechanism.  Let $x_{ij}^d(t)$ be  \emph{data plane decision variable} that represents 
the transmission rate assigned by the data plane to commodity $d$ on link $(i,j)$ for slot $t$. These are chosen to satisfy
\begin{align}
  x_{ij}^d(t,T) & = \sum_{\tau = t}^{t+T-1} x_{ij}^d(\tau) &&\hspace{-1mm}\text{for}~i\in \set{S}, j \in \set{H}_i, t \in \set{T}   \label{eq:xijdT}\\  
  \sum_{d \in \set{D}_{ij}} x_{ij}^d(t) & \leq c_{ij} &&\hspace{-1mm}\text{for}~i\in \set{S}, j \in \set{H}_i, t \in \IntSet_+. \label{eq:cij} 
\end{align}

\subsection{Queues}

\label{sec:Queue}

Packets are queued at each switch according to their commodity.  Let  $Q_i^d(t)$ be the number of commodity-$d$ 
packets queued at switch $i$ on slot $t$.  The value $Q_i^d(t)$ is also called  the \emph{commodity-$d$ backlog}
and satisfies:
\begin{multline}
  \label{eq:Qid}
  Q_i^d(t+1) \leq \prts{ Q_i^d(t) - \sum_{j \in \set{H}_i^d} x_{ij}^d(t) - b_i^d }_+ +\\\sum_{j \in \set{P}_i^d} x_{ji}^d(t) + a_i^d(t) \quad\text{for}~ i \in \set{S}, d \in \set{D},
\end{multline}
where $\prts{z}_+ = \max[0,z]$.  Note that $\sum_{j \in \set{H}_i^d} x_{ij}^d(t)$ denotes  output transmission rate to next-hop switches, and $\sum_{j \in \set{P}_i^d} x_{ji}^d(t)$ denotes receiving transmission rate from previous-hop switches.  
Inequality \eqref{eq:Qid} is an inequality rather than an equality because the actual amount of new endogenous arrivals
on slot $t$ may be less than $\sum_{j \in \set{P}_i^d} x_{ji}^d(t)$ if previous switches $j$ do not have enough commodity $d$
backlog to fill the assigned transmission rate $x_{ji}^d(t)$. 
It can be shown that the backlog at time $t \in \set{T}$ satisfies for every $i \in \set{S}, d \in \set{D}$:
\begin{multline}
  \label{eq:QidT}
  Q_i^d(t+T) \leq \prts{ Q_i^d(t) - \sum_{j \in \set{H}_i^d} x_{ij}^d(t,T) - T b_i^d }_+ + \\\sum_{j \in \set{P}_i^d} x_{ji}^d(t,T) + \sum_{\tau=t}^{t+T-1} a_i^d(\tau).
\end{multline}

Note that, while a common queue for each commodity is not available in practical switches, it can be heuristically implemented by queues in an SDN switch in Section \ref{sec:Implementation}.

\subsection{Stability and Assumption}
\begin{definition}[Queue Stability  \cite{Neely:SNO}]   
  \label{def:QueueStability}
  A queue with backlog $\prtc{ Z(t) \geq 0 : t \in \IntSet_+ }$ is strongly stable if
\begin{equation*}
  \limsup_{t \rightarrow \infty} \frac{1}{t} \sum_{\tau = 0}^{t - 1} \expect{ Z(\tau) } < \infty.
\end{equation*}
\end{definition}

\begin{definition}[Network Stability  \cite{Neely:SNO}] 
  A network is strongly stable when every queue in the network is strongly stable.
\end{definition}

The arrival and departure rates are assumed to satisfy a standard Slater condition: 
\begin{assumption}[Slater Condition]
  \label{ass:Slater}
  There exists an $\epsilon>0$ and a randomized policy $\prtc{ x_{ij}^{d\ast}(t) }_{i \in \set{S}, d \in \set{D}, j \in \set{H}_i^d}$ with
\begin{multline*}
  \expect{ \sum_{j \in \set{P}_i^d} x_{ji}^{d\ast}(t) + a_i^d(t) - \sum_{j \in \set{H}_i^d} x_{ij}^{d\ast}(t) - b_i^d } < -\epsilon \\\text{for all}~ i \in \set{S}, d \in \set{D}, t \in \IntSet_+,
\end{multline*}
and the randomized policy satisfies the constraint \eqref{eq:cij}.
\end{assumption}

Note that Assumption \ref{ass:Slater} is stated in terms of only the data plane decision variables $x_{ij}^{d*}(t)$ and
the constraint \eqref{eq:cij}.  While the control plane decisions are used in the algorithm and the additional 
constraint \eqref{eq:xijdT} is satisfied by the algorithm, those  are not used in the 
Slater condition. 

\section{Throughput-Optimal Algorithm}
\label{sec:IdealAlgorithm}
\subsection{The Algorithm}
The novel throughput-optimal algorithm runs distributively at every switch.  Let $K$ be a positive real number, and $x_{ij}^d(-T,T) = 0$ for all $i,j \in \set{S}, d \in \set{D}$.  At every reconfiguration time $t \in \set{T}$, switch $i$ executes Algorithm \ref{alg:Main} for each link connecting to a next-hop switch $j$ for $i \in \set{S}, j \in \set{H}_i$.   Recall that $x_{ij}^d(t,T) = 0$ for $d \in \set{D} \backslash \set{D}_{ij}$.  The function $\round{z}$ rounds the real-valued $z$ to its closest integer.
  
\begin{algorithm}
  \begin{algorithmic}
    \STATE{ \texttt{// switch $i$, link to switch $j$, time $t$ //} }\vspace{1mm}
    \STATE{ $y_{ij}^d(t) \leftarrow Q_i^d(t) - Q_j^d(t) + x_{ij}^d(t-T,T) $ for $d \in \set{D}_{ij}$ }
    \STATE{ $k_{ij}(t) \leftarrow \max\prts{ 1, \min\prts{ K, \frac{1}{Tc_{ij}} \sum_{d \in \set{D}_{ij}} \prts{ y_{ij}^d(t) }_+ } }$ }      
    \IF{ $\sum_{d \in \set{D}_{ij}} \round{ \prts{ y_{ij}^d(t) }_+ / k_{ij}(t) } \leq T c_{ij}$ }
    \STATE{ $x_{ij}^d(t,T) \leftarrow \round{ \prts{ y_{ij}^d(t) }_+ / k_{ij}(t) }$ for $d \in \set{D}_{ij}$ }
    \RETURN{ $\prtc{ x_{ij}^d(t,T) }_{d \in \set{D}_{ij}}$ }
    \ELSE
      \RETURN{ result of Algorithm \ref{alg:PacketFilling} }
    \ENDIF
  \end{algorithmic}
  \caption{Throughput-optimal rate allocation}
  \label{alg:Main}  
\end{algorithm}
\begin{algorithm}
  \begin{algorithmic}
    \STATE{ \texttt{// switch $i$, link to switch $j$, time $t$ //} }\vspace{1mm}    
    \STATE{ $y_{ij}^d(t) \leftarrow Q_i^d(t) - Q_j^d(t) + x_{ij}^d(t-T,T) $ for $d \in \set{D}_{ij}$ }    
    \STATE{ $v_{ij}^d \leftarrow 0$ for $d \in \set{D}_{ij}$ }
    \FOR{ $n = 1$ to $T c_{ij}$ }
    \STATE{ $d_n \leftarrow \argmax_{e \in \set{D}_{ij}} \prts{ y_{ij}^e(t) - v_{ij}^e k_{ij}(t) }$ }
    \STATE{ $v_{ij}^{d_n} \leftarrow v_{ij}^{d_n} + 1$ }
    \ENDFOR
    \STATE{ $x_{ij}^d(t,T) \leftarrow v_{ij}^d$ for $d \in \set{D}_{ij}$ }    
    \RETURN{ $\prtc{ x_{ij}^d(t,T) }_{d \in \set{D}_{ij}}$ }
  \end{algorithmic}
  \caption{Packet-filling algorithm (unaccelerated version)}
  \label{alg:PacketFilling}
\end{algorithm}

\subsection{Intuitions}
\label{sec:Intuitions}
Algorithm \ref{alg:Main} solves problem \eqref{eq:MainProblem} with the value of $k_{ij}(t)$ that depends on local queue information and previous decisions.  This $k_{ij}(t)$ is deliberately introduced, just as a solution of problem \eqref{eq:MainProblem} imitates weighted fair queueing, which provides fairness and low latency in practice \cite{WFQ,GPS}.
\begin{align}
  \minimize \quad & \sum_{d \in \set{D}_{ij}} \Bigg\{ x_{ij}^d(t, T) \prts{ Q_j^d(t) - Q_i^d(t) } + \Bigg. \notag \\ &\quad \Bigg. \frac{ k_{ij}(t) }{ 2 } \prts{ x_{ij}^d(t, T) - \frac{ x_{ij}^d(t-T, T) }{ k_{ij}(t) } }^2 \Bigg\} \label{eq:MainProblem} \\
  \subjectto \quad & \sum_{d \in \set{D}_{ij}} x_{ij}^d(t, T) \leq T c_{ij} \notag\\
  & x_{ij}^d(t, T) \in \IntSet_+ \hspace{4.5em} \text{for all}~~ d \in \set{D}_{ij} \notag
\end{align}

In Algorithm \ref{alg:Main}, $y_{ij}^d(t)$ represents a request for transmission rate of commodity $d$.  It also indicates how well the previous rates were allocated.  Under allocation of $x_{ij}^d(t-T,T)$ increases queue backlog $Q_i^d(t)$, which tends to increase the request $y_{ij}^d(t)$.  It is easy to see that, if the queue backlogs $Q_i^d(t)$ and $Q_j^d(t)$ are about the same, then the request is about the same as its previous value.  This behavior is smoother than MaxWeight.

The requests are fulfilled in two situations.  i) When the total requests are roughly within $K T c_{ij}$, i.e., $\sum_{d \in \set{D}_{ij}} \round{ \prts{ y_{ij}^d(t) }_+ / k_{ij}(t) } \leq T c_{ij}$, the requests are fulfilled in WFQ fashion, which can be seen by considering $k_{ij}(t) = \frac{1}{T c_{ij}} \sum_{d \in \set{D}_{ij}} \prts{ y_{ij}^d(t) }_+$ and
\begin{align}
  x_{ij}^d(t, T) & = \round{ \prts{ y_{ij}^d(t) }_+ / k_{ij}(t) } \notag\\
    & = \round{ \frac{ \prts{ y_{ij}^d(t) }_+ }{ \sum_{e \in \set{D}_{ij}} \prts{ y_{ij}^e(t) }_+ } \times T c_{ij} }. \label{eq:WFQIntuition}
\end{align}
ii) The other is an extreme situation for stability analysis, which occurs when a network operates near its capacity, which may not be the case in practice due to TCP congestion control.  This case is solved by Algorithm \ref{alg:PacketFilling}\footnote{Algorithm \ref{alg:PacketFilling} can be accelerated by fulfilling multiple requests per iteration. For example, the requests in Fig. \ref{fig:PacketFilling} can be fulfilled in three iterations.}, as illustrated by Fig. \ref{fig:PacketFilling}.  It is easy to see the fairness introduced by $k_{ij}(t)$.  Without $k_{ij}(t)$, i.e., $k_{ij}(t)$ is always $1$, the allocation in Fig. \ref{fig:PacketFilling} will be $0, 0, 2, 7$ for commodities $1$ to $4$, which may cause a fairness issue with TCP flows \cite{TCP-Backpressure}.

\begin{figure}
  \centering
  \includegraphics[scale=1.0]{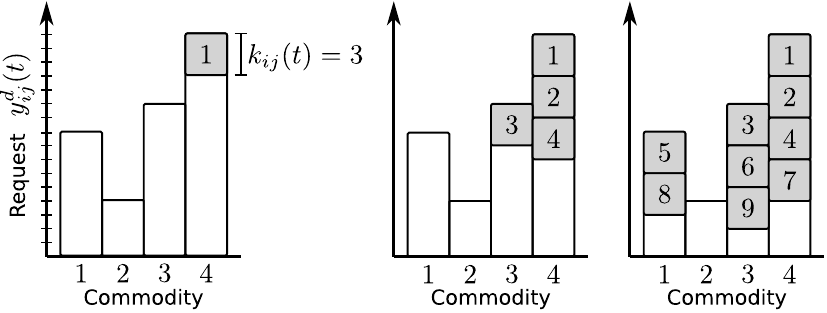}\vspace{-2mm}
  \caption{Packet-filling Algorithm \ref{alg:PacketFilling} iteratively fulfills the requests.  An iteration number is indicated in a gray box.  In this example, $T c_{ij} = 9$ and the first iteration (the plot on the left) allocates rate to commodity 4.  The algorithm allocates $2, 0, 3, 4$ rates to commodities $1, 2, 3,4$.}\vspace{-4.5mm}
  \label{fig:PacketFilling}
\end{figure}

\subsection{Correctness of Algorithm \ref{alg:Main}}
\label{sec:solution}
This subsection shows that  Algorithm \ref{alg:Main} returns an optimal solution of problem \eqref{eq:MainProblem}.  To simplify notation in this section, the time index of variables and constants in problem \eqref{eq:MainProblem} are omitted. Let $z_{ij}^d$ denote $x_{ij}^d(t-T,T)$.
    
When commodity $d$ is allocated an integer service value $v_{ij}^d$, define its contribution to the cost function of problem \eqref{eq:MainProblem} as
\begin{equation*}
  \label{eq:gijd}
  g_{ij}^d( v_{ij}^d ) = v_{ij}^d \prts{ Q_j^d - Q_i^d } + \frac{ k_{ij} }{ 2 } \prts{ v_{ij}^d - \frac{ z_{ij}^d }{ k_{ij} } }^2.
\end{equation*}
The cost difference of getting another service allocation is
\begin{equation}
  \label{eq:diffgijd}
g_{ij}^d( v_{ij}^d + 1 ) - g_{ij}^d( v_{ij}^d ) = - \prts{ Q_i^d - Q_j^d + z_{ij}^d - \frac{ k_{ij} }{ 2 } - k_{ij} v_{ij}^d }.
\end{equation}
Since the cost function in problem \eqref{eq:MainProblem} is minimized, commodity $d$ only accepts a transmission allocation if $g_{ij}^d( v_{ij}^d + 1 ) \leq g_{ij}^d( v_{ij}^d )$.  The cost difference in \eqref{eq:diffgijd} is monototically increasing in $v_{ij}^d$.  Therefore, commodity $d$ receives transmission allocation at most:
\begin{align}
  x_{ij}^{d\maxvar} &= \min \prtc{ v \in \IntSet_+ : g_{ij}^d( v + 1 ) - g_{ij}^d( v ) > 0 } \label{eq:xijdmaxDef}\\
  & = \min \prtc{ v \in \IntSet_+ : \frac{ Q_i^d - Q_j^d + z_{ij}^d }{ k_{ij} } - 0.5 < v } \notag\\
  & = \round{ \prts{ Q_i^d - Q_j^d + z_{ij}^d }_+ / k_{ij} }. \label{eq:xijdmax}
\end{align}

\begin{lemma}
  \label{lem:solution2}
  When $k_{ij} > 0$ and $\sum_{d \in \set{D}_{ij}} x_{ij}^{d\maxvar} \leq T c_{ij}$, the optimal solution of problem \eqref{eq:MainProblem} is $x_{ij}^d = x_{ij}^{d\maxvar}$ for $d \in \set{D}_{ij}$.
  \begin{proof}
    For any $v_{ij}^d \in \prtc{ 0, 1, \dotsc, x_{ij}^{d\maxvar} }, d \in \set{D}_{ij}$, the given implies that $\sum_{d \in \set{D}_{ij}} v_{ij}^d \leq \sum_{d \in \set{D}_{ij}} x_{ij}^{d\maxvar} \leq T c_{ij}$.  So, any chosen $v_{ij}^d$ in $\prtc{ 0, 1, \dotsc, x_{ij}^{d\maxvar} }$ leads to a feasible solution of problem \eqref{eq:MainProblem}.  Note that the objective function of problem \eqref{eq:MainProblem} is separable, $\sum_{d \in \set{D}_{ij}} g_{ij}^d( v_{ij}^d )$, and is minimized.  The definition of $x_{ij}^{d\maxvar}$ in \eqref{eq:xijdmaxDef} implies $g_{ij}^d( v ) > g_{ij}^d( x_{ij}^{d\maxvar} )$ for any $v > x_{ij}^{d\maxvar}$, so any $v > x_{ij}^{d\maxvar}$ is not optimal.  Thus, i) if $x_{ij}^{d\maxvar} = 0$, $x_{ij}^d = 0 = x_{ij}^{d\maxvar}$ minimizes problem \eqref{eq:MainProblem} with respect to commodity $d$.  ii) If $x_{ij}^{d\maxvar} > 0$, the definition of $x_{ij}^{d\maxvar}$ in \eqref{eq:xijdmaxDef} implies $g_{ij}^d( x_{ij}^{d\maxvar} ) \leq g_{ij}^d( v )$ for any $v \in \prtc{ 0, 1, \dotsc, x_{ij}^{d\maxvar} - 1 }$, so $x_{ij}^d = x_{ij}^{d\maxvar}$ minimizes the problem with respect to commodity $d$.
  \end{proof}
\end{lemma}

Lemma \ref{lem:solution2} implies that Algorithm \ref{alg:Main} solves problem \eqref{eq:MainProblem} when the first (if-)condition is met.  The other case of Algorithm \ref{alg:Main} can be shown by the following property.  When commodity $d$ is allocated $v_{ij}^d$ transmission rate, define the unfulfilled level of commodity-$d$'s request as
\begin{equation}
  \label{eq:lijd}
  l_{ij}^d( v_{ij}^d ) = Q_i^d - Q_j^d + z_{ij}^d - k_{ij} v_{ij}^d.
\end{equation}

\begin{lemma}
  \label{lem:level}
  When $k_{ij} > 0$, for any commodities $d, e \in \set{D}_{ij}$ whose allocated rates are respectively $v_{ij}^d$ and $v_{ij}^e$, the following holds:\\
  i) If $l_{ij}^d ( v_{ij}^d ) = l_{ij}^e ( v_{ij}^e )$, then
  \begin{equation*}
    g_{ij}^d( v_{ij}^d + 1 ) + g_{ij}^e( v_{ij}^e ) = g_{ij}^d( v_{ij}^d ) + g_{ij}^e( v_{ij}^e + 1 ).
  \end{equation*}
  ii) If $l_{ij}^d ( v_{ij}^d ) > l_{ij}^e ( v_{ij}^e )$, then
  \begin{equation*}
    g_{ij}^d( v_{ij}^d + 1 ) + g_{ij}^e( v_{ij}^e ) < g_{ij}^d( v_{ij}^d ) + g_{ij}^e( v_{ij}^e + 1 ).
  \end{equation*}
  \begin{proof}
    It holds from equations \eqref{eq:diffgijd} and \eqref{eq:lijd} that
    \begin{align}
      & g_{ij}^d( v_{ij}^d + 1 ) - g_{ij}^d ( v_{ij}^d ) - g_{ij}^e( v_{ij}^e + 1 ) + g_{ij}^e( v_{ij}^e ) \notag\\
      & \quad = - \prts{ l_{ij}^d( v_{ij}^d ) - \frac{ k_{ij} }{ 2 } } + \prts{ l_{ij}^e( v_{ij}^e ) - \frac{ k_{ij} }{ 2 } } \notag\\
      & \quad = - l_{ij}^d( v_{ij}^d ) + l_{ij}^e( v_{ij}^e ). \label{eq:gg}
    \end{align}
    In case (i), $l_{ij}^d ( v_{ij}^d ) = l_{ij}^e ( v_{ij}^e )$ implies that $g_{ij}^d( v_{ij}^d + 1 ) - g_{ij}^d ( v_{ij}^d ) - g_{ij}^e( v_{ij}^e + 1 ) + g_{ij}^e( v_{ij}^e ) = 0$, which proves (i).  Case (ii) can be proven similarly by substituting $- l_{ij}^d ( v_{ij}^d ) + l_{ij}^e ( v_{ij}^e ) < 0$ into equation \eqref{eq:gg} and rearranging terms.
  \end{proof}
\end{lemma}

Lemma \ref{lem:level} implies that allocating rate to the commodity with the highest unfulfilled level reduces the total objective the most.  Specifically, let $v_{ij}^d$ be the current rate allocation of commodity $d$.  For $d^\ast = \argmax_{d \in \set{D}_{ij}} l_{ij}^d(v_{ij}^d)$, Lemma \ref{lem:level} implies that
\begin{multline*}
  g_{ij}^{d^\ast}( v_{ij}^{d^\ast} + 1 ) + \sum_{d \in \set{D}_{ij}\backslash\prtc{d^\ast}} g_{ij}^d( v_{ij}^d ) \leq \\g_{ij}^{e}( v_{ij}^{e} + 1 ) + \sum_{d \in \set{D}_{ij}\backslash\prtc{e}} g_{ij}^d( v_{ij}^d ) \quad \text{for all}~ e \in \set{D}_{ij}.
\end{multline*}

The above property ensures that the iterative allocation in Algorithm \ref{alg:PacketFilling} greedily optimizes problem \eqref{eq:MainProblem} when event $\sum_{d \in \set{D}_{ij}} \round{ \prts{ y_{ij}^d(t) }_+ / k_{ij}(t) } > T c_{ij}$ occurs.  It can be proven by contradiction that commodity $d$ gets at most $x_{ij}^{d\maxvar}$ rate for all $d \in \set{D}_{ij}$.  Then, the algorithm always allocates the entire transmission rate, as an implication of the event.  Therefore, Algorithm \ref{alg:PacketFilling} returns an optimal solution of problem \eqref{eq:MainProblem}.

\begin{theorem}
  Given $K > 0$, Algorithm \ref{alg:Main} solves problem \ref{eq:MainProblem} with $k_{ij}(t) \in [1, K]$, where $k_{ij}(t)$ is defined in the algorithm.
\begin{proof}
The theorem is the consequence of Lemmas \ref{lem:solution2} and \ref{lem:level} and the fact that $k_{ij}(t) \in [1, K]$.
\end{proof}
\end{theorem}

\subsection{Stability Analysis}
Problem \ref{eq:MainProblem} with $k_{ij}(t) \in [1, K]$ is shown to be a class of throughput-optimal policies.  Let $Q(t) = \prtr{ Q_i^d(t) }_{i \in \set{S}, d \in \set{D}}$ be a vector of all backlogs at time $t$.  Define $\norm{z}_1$ as the $l_1$-norm, e.g., $\norm{ Q(t) }_1 = \sum_{i \in \set{S}} \sum_{d \in \set{D}} Q_i^d(t)$.

\begin{theorem}
  When Assumption \ref{ass:Slater} holds, the network is strongly stable:
  \begin{equation*}
    \limsup_{N \rightarrow \infty} \frac{1}{NT} \sum_{n=0}^{N-1}\sum_{\tau=0}^{T-1} \expect{ \norm{Q(nT + \tau)}_1 } \leq \frac{G_1}{\epsilon} + G_2,
  \end{equation*}
  where $G_1 = \abs{\set{S}} \abs{\set{D}} \prts{ T \delta^2 \prtr{ \abs{\set{S}} + 1 } + 2 K T \delta^2 \abs{\set{S}} }$ and $G_2 = (T-1)\delta(\abs{\set{S}} + 1)/2$.
\end{theorem}

\begin{proof}
Squaring both sides of \eqref{eq:QidT}, rearranging, and bounding terms (see \cite{Neely:SNO} for example) leads to
\begin{multline*}
  \frac{1}{2} \prts{ Q_i^d(t+T)^2 - Q_i^d(t)^2 } \leq Q_i^d(t) \prts{ \sum_{\tau=t}^{t+T-1} a_i^d(\tau) - Tb_i^d }\\+Q_i^d(t) \Bigg[ \sum_{j \in \set{P}_i^d} x_{ji}^d(t,T) - \sum_{j \in \set{H}_i^d} x_{ij}^d(t,T) \Bigg] + C_i^d,
\end{multline*}
where $C_i^d = \frac{ T^2 \delta^2 \prtr{ \abs{ \set{P}_i^d } + \abs{ \set{H}_i^d } + 2 } }{2}$.  Define a $T$-slot quadratic Lyapunov drift of queue backlogs \cite{Neely:SNO} as
\begin{equation*}
  \Delta(t, T) = \frac{1}{2} \prts{ \norm{ Q(t+T) }^2 - \norm{ Q(t) }^2 },
\end{equation*}
where $\norm{z}$ is the $l_2$-norm of vector $z$, i.e., $\norm{Q(t)}^2 = \sum_{i \in \set{S}}\sum_{d \in \set{D}} Q_i^d(t)^2$. It holds that
\begin{multline}
  \label{eq:DriftBound}
  \Delta(t, T) \leq \sum_{i \in \set{S}} \sum_{d \in \set{D}} \prtc{ C_i^d + Q_i^d(t) \prts{ \sum_{\tau=t}^{t+T-1} a_i^d(\tau) - Tb_i^d } } \\ + \sum_{i \in \set{S}} \sum_{d \in \set{D}} Q_i^d(t) \Bigg[ \sum_{j \in \set{P}_i^d} x_{ji}^d(t, T) - \sum_{j \in \set{H}_i^d} x_{ij}^d(t, T) \Bigg]
\end{multline}
The second line of the above equation can be rewritten as
\begin{multline}
  \label{eq:Rearrange}
\sum_{i \in \set{S}} \sum_{d \in \set{D}} Q_i^d(t) \Bigg[ \sum_{j \in \set{P}_i^d} x_{ji}^d(t, T) - \sum_{j \in \set{H}_i^d} x_{ij}^d(t, T) \Bigg] \\= \sum_{i \in \set{S}} \sum_{j \in \set{H}_i} \sum_{d \in \set{D}_{ij}} x_{ij}^d(t, T) \prts{ Q_j^d(t) - Q_i^d(t) },
\end{multline}
using the fact that $x_{ij}^d(t, T) = 0$ for every $d \in \set{D}\backslash\set{D}_{ij}$.

  Instead of minimizing the above expression, which leads to the MaxWeight algorithm, an state-dependent proximal term
\begin{equation*}
  \frac{k_{ij}(t)}{2} \prts{ x_{ij}^d(t, T) - \frac{ x_{ij}^d(t-T,T) }{ k_{ij}(t) } }^2 \quad \text{with}
\end{equation*}
\begin{equation*}
  k_{ij}(t) = \prts{ \frac{ 1 }{T c_{ij}} \sum_{d \in \set{D}_{ij}} \prts{ Q_i^d(t) - Q_j^d(t) + x_{ij}^d(t-T, T) }_+ }_{[1,K]}
\end{equation*}
 is introduced\footnote{The $x_{ij}^d(t-T, T)$ in the proximal term can be replaced by $z_{ij}^d(t) / \alpha$ for any $0 \leq z_{ij}^d(t) < \infty$ and $0 < \alpha < \infty$ to allow more control flexibility.}, where $\prts{z}_{\prts{1,K}} = \max[1, \min[ K, z ] ]$.  This proximal term is non-negative and is upper bounded by $K T^2 \delta^2$, so it holds from \eqref{eq:DriftBound} and \eqref{eq:Rearrange} that
\begin{align}
  \Delta(t, T)
  & \leq \sum_{i \in \set{S}} \sum_{d \in \set{D}} \prtc{ C_i^d + Q_i^d(t) \prts{ \sum_{\tau=t}^{t+T-1} a_i^d(\tau) - Tb_i^d } } \notag\\
  & \quad + \sum_{i \in \set{S}} \sum_{j \in \set{H}_i} \sum_{d \in \set{D}_{ij}} \Bigg\{ x_{ij}^d(t, T) \prts{ Q_j^d(t) - Q_i^d(t) } \notag\\
  & \quad + \frac{k_{ij}(t)}{2} \prts{ x_{ij}^d(t, T) - \frac{ x_{ij}^d(t-T,T) }{ k_{ij}(t) } }^2 \Bigg\}.   \label{eq:DriftMin}
\end{align}
Minimizing the right-hand-side of \eqref{eq:DriftMin} with respect to $\prtc{x_{ij}^d(t,T)}_{d \in \set{D}_{ij}}$ leads to problem \eqref{eq:MainProblem}.  Applying the result from Algorithm \ref{alg:Main}, which solves the minimization at reconfiguration time $t \in \set{T}$, yields the bound for any other $\prtc{ \hat{x}_{ij}^d(t, T) }_{d \in \set{D}_{ij}}$ satisfying constraints in problem \eqref{eq:MainProblem}:
\begin{align*}
  \Delta(t, T)
  & \leq \sum_{i \in \set{S}} \sum_{d \in \set{D}} \prtc{ C_i^d + Q_i^d(t) \prts{ \sum_{\tau=t}^{t+T-1} a_i^d(\tau) - Tb_i^d } } \\
  & \quad + \sum_{i \in \set{S}} \sum_{j \in \set{H}_i} \sum_{d \in \set{D}_{ij}} \Bigg\{ \hat{x}_{ij}^d(t, T) \prts{ Q_j^d(t) - Q_i^d(t) } \\
  & \quad + \frac{k_{ij}(t)}{2} \prts{ \hat{x}_{ij}^d(t, T) - \frac{ x_{ij}^d(t-T, T) }{ k_{ij}(t) } }^2 \Bigg\}
\end{align*}
Since the proximal term is bounded and policy $\prtc{ x_{ij}^{d\ast}(t, T) = \prtc{ x_{ij}^{d\ast}(\tau) }_{\tau = t }^{t+T-1} }_{d \in \set{D}_{ij}}$, constructed from the randomized policy in Assumption \ref{ass:Slater}, is one of those $\prtc{ \hat{x}_{ij}^d(t, T) }_{d \in \set{D}_{ij}}$, it follows that
\begin{align*}
  \Delta(t, T)
  & \leq \sum_{i \in \set{S}} \sum_{d \in \set{D}} \prtc{ C_i^d + Q_i^d(t) \prts{ \sum_{\tau=t}^{t+T-1} a_i^d(\tau) - Tb_i^d } } \\
  & \hspace{-3em}+ \sum_{i \in \set{S}} \sum_{j \in \set{H}_i} \sum_{d \in \set{D}_{ij}} \Bigg\{ x_{ij}^{d\ast}(t, T) \prts{ Q_j^d(t) - Q_i^d(t) } + K T^2 \delta^2 \Bigg\}
\end{align*}
Applying identity \eqref{eq:Rearrange}, taking expectation, and using the independent property of the randomized policy gives
\begin{multline*}
  \expect{ \Delta(t, T) } \leq \sum_{i \in \set{S}} \sum_{d \in \set{D}} \Bigg\{ D_i^d + \expect{ Q_i^d(t) } \times \\
  \sum_{\tau = t}^{t+T-1} \expect{ \sum_{j \in \set{P}_i^d} x_{ji}^{d\ast}(\tau) + a_i^d(\tau) - \sum_{j \in \set{H}_i^d} x_{ij}^{d\ast}(\tau) - b_i^d } \Bigg\},
\end{multline*}
where $D_i^d = C_i^d + K T^2 \delta^2\prtr{ \abs{\set{P}_i^d} + \abs{\set{H}_i^d} }$.  Assumption \ref{ass:Slater} implies for all $t \in \set{T}$ that
\begin{equation}
  \label{eq:DriftNeg}
  \frac{1}{2T} \expect{ \norm{ Q(t+T) }^2 - \norm{ Q(t) }^2 } \leq G_1 - \epsilon \sum_{i \in \set{S}} \sum_{d \in \set{D}} \expect{ Q_i^d(t) },
\end{equation}
where $G_1$ is defined in the theorem.

Queue dynamic \eqref{eq:Qid} and the upper bound $\delta$ imply that $Q_i^d(t + \tau) \leq Q_i^d( t ) + \tau \delta ( \abs{ \set{S} } + 1 )$ for any $\tau \in \IntSet_+$ and $i \in \set{S}, d \in \set{D}$.  Summing for $\tau \in \prtc{0, 1, \dotsc, T-1}$ gives $\sum_{\tau = 0}^{T-1} Q_i^d(t+\tau) \leq T Q_i^d(t) + T (T - 1) \delta ( \abs{\set{S}} + 1 )/2$ and
\begin{equation*}
  Q_i^d(t) \geq \frac{1}{T} \sum_{\tau=0}^{T-1} Q_i^d(t+\tau) - G_2 \quad \text{for}~ t \in \IntSet_+,
\end{equation*}
where $G_2$ is defined in the theorem.  Substituting the above into inequality \eqref{eq:DriftNeg} yields:
\begin{multline*}
  \label{eq:DriftNeg}
  \frac{1}{2T} \expect{ \norm{ Q(t+T) }^2 - \norm{ Q(t) }^2 } \leq G_1 + \epsilon G_2  \\ - \frac{\epsilon}{T} \sum_{\tau=0}^{T-1} \expect{ \norm{ Q(t+\tau) }_1 }  \quad\text{for}~ t \in \set{T}.
\end{multline*}
Telescope summation for $t \in \prtc{0, T, \dotsc, (N-1)T}$ gives:
\begin{multline*}
  \frac{1}{2T} \expect{ \norm{ Q(NT) }^2 - \norm{ Q(0) }^2 } \leq G_1 N + \epsilon G_2 N \\- \frac{\epsilon}{T} \sum_{n=0}^{N-1} \sum_{\tau=0}^{T-1} \expect{ \norm{ Q(nT+\tau) }_1 }.
\end{multline*}
Rearranging terms and taking supremum limit as $N \rightarrow \infty$ proves the theorem.
\end{proof}

\section{System Realization}
\label{sec:Implementation}

The previous section provides an ideal allocation of decision variables.  This section develops a heuristic improvement
that is easier to implement for practical switches. 

\subsection{Approximation of Common Queues}
\label{sec:QueueApprox}
An SDN switch has output queues at each of its ports\cite{DeTail,DiffMax,OpenFlow}.  Those queues can be assigned to each commodity.  Let $Q_{ij}^d(t)$ denote the backlog of a queue for commodity $d$ at the port of switch $i$ connecting to switch $j$ at time $t$.  The queue backlog $Q_i^d(t)$ in Section \ref{sec:Queue} can be approximated by $\tilde{Q}_i^d(t) = \sum_{j \in \set{H}_i^d} Q_{ij}^d(t)$, for $i \in \set{S}, d \in \set{D}$.  It can be shown that this approximation becomes exact when the port's queues have never been emptied.  Note that OpenFlow \cite{OpenFlow} allows $2^{32}$ unique queues per port, but availability of those queues may depend on a switch.  This work encourages next generation switches to have a high number of available queues.

\subsection{Additional Packet Headers}
Three fields are appended into the IP header as IP options: \texttt{CommodityId}, \texttt{QueueInfo}, and \texttt{HashField}.  The \texttt{CommodityId} identifies the commodity of \texttt{QueueInfo} which stores the rounded value of exponential moving average of approximated queue backlog, $\tilde{Q}_j^d(t)$.  A packet from switch $j$ to switch $i$ carries queue information of a commodity, which is circularly selected from commodities in $\set{D}_{ij}$\footnote{This round-robin technique is inspired by CONGA \cite{CONGA}.  The concept can also be applied to VXLAN \cite{VXLAN}.}.  The \texttt{HashField} is used for traffic splitting and is explained in Section \ref{sec:Hashing}.

Once a packet with the additional headers from switch $j$ arrives to switch $i$, the contained queue information is extracted and stored in a local memory, which is denoted by $M_{ij}^d(t)$. This is the most recent queue information for commodity $d$ on link $(i,j)$ up to time $t$, where $d$ is the value in  \texttt{CommodityId}.  The header processing can be implemented by P4 \cite{P4}, NetFPGA \cite{NetFPGA}, or a custom ASIC.


\subsection{Weighted Fair Queueing}
Each port of switch $i$ connecting to switch $j \in \set{H}_i$ is configured with weighted fair queueing.  Let $r_{ij}^d(t, T)$ denote the measured number of commodity-$d$ packets transmitted from switch $i$ to switch $j$ during the interval $[t, t+T)$.  At every configuration time $t \in \set{T}$, the weight $w_{ij}^d(t)$ for the interval $[t,t+T)$ is
\begin{equation*}
  w_{ij}^d(t) = \max\prts{ 1, \tilde{Q}_i^d(t) - M_{ij}^d(t) + r_{ij}^d(t-T,T)/\alpha }, d \in \set{D}_{ij}
\end{equation*}
and $w_{ij}^d(t) = 0$ for $d \in \set{D}\backslash\set{D}_{ij}$.  Parameter $\alpha > 0$ is added to scale the magnitude of actual traffic to match a finite queue capacity.  From these weights, $w_{ij}^d(t)/\sum_{e \in \set{D}} w_{ij}^e(t)$ fraction of link capacity is given to commodity $d$, which corresponds to the intuition in equation \eqref{eq:WFQIntuition}.  Note that actual traffic  $r_{ij}^d(t-T,T)$ is used instead of $x_{ij}^d(t-T,T)$, because it is a better approximation under TCP traffic.

\subsection{Traffic Splitting by Hashing}
\label{sec:Hashing}
A sending rate of a TCP connection is reduced when out-of-order packets are received at a destination.  Hashing a packet to a next-hop switch based on \texttt{HashField} is implemented to reduce packet reordering.  The \texttt{HashField} is generated once for the entire packet life.  Packets from the same TCP connection have the same \texttt{HashField}, so they are hashed to the same path.  Reordering does not occur if a hash rule at each switch is the same for the entire TCP connection.  The hash rule is calculated as follows.  For each commodity $d \in \set{D}$, define $\prtc{ s_{ij}^d(t) }_{j \in \set{H}_i^d}$ as a solution of
\begin{align}
  \maximize &\quad \min_{j \in \set{H}_i^d} \prtc{ Q_{ij}^d(t) - r_{ij}^d(t-T,T) + s_{ij}^d(t) } \label{eq:MaxMin}\\
    \subjectto &\quad \sum_{j \in \set{H}_i^d} s_{ij}^d(t) = \sum_{j \in \set{H}_i^d} r_{ij}^d(t-T,T) \notag\\
      & \quad s_{ij}^d(t) \in \IntSet_+ \quad \text{for all}~ j \in \set{H}_i^d. \notag
\end{align}
This problem can be solved in polynomial time.  Let $s_{ij}^d(t) = 0$ for all $j \in \set{H}_i^d$.  Iteratively, $s_{ij}^d(t)$ is increased for a group of indices in $\Argmin_{j \in \set{H}_i^d} \prts{ Q_{ij}^d(t) - r_{ij}^d(t-T,T) + s_{ij}^d(t) }$ until the equality constraint is met\footnote{The process can be accelerated by increasing multiple values of $s_{ij}^d(t)$'s per iteration.}.  This process keeps increasing the $\min_{j \in \set{H}_i^d} \prtc{ Q_{ij}^d(t) - r_{ij}^d(t-T,T) + s_{ij}^d(t) }$.  The intuition of problem \eqref{eq:MaxMin} is that it attempts to equalize the backlog levels at all ports at end of the interval $[t,t+T)$ by using $r_{ij}^d(t-T,T)$ as an estimate of actual transmission $r_{ij}^d(t,T)$.  The attempt tries to make the approximation in Section \ref{sec:QueueApprox} exact.  The splitting ratio of the port connecting to switch $j$ is $f_{ij}^d(t) = s_{ij}^d(t) / \sum_{k \in \set{H}_i^d} s_{ik}^d(t)$ for $j \in \set{H}_i^d$.  

\section{Simulations}
\label{sec:Simulations}

\subsection{Ideal Simulation}
Algorithm \ref{alg:Main} is simulated according to the system model in Section \ref{sec:Model}.  A switch uses a token bucket mechanism to ensure that transmission rate per interval satisfy constraint \eqref{eq:xijdT} after $x_{ij}^d(t, T)$ is determined.

\begin{figure}
  \centering
  \includegraphics[scale=1.0]{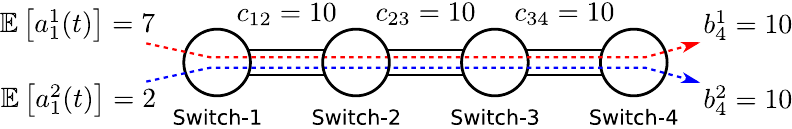}\vspace{-2mm}
  \caption{Line Network with $\set{S}=\prtc{1,2,3,4}, \set{D}=\prtc{1,2}$ and $\set{H}_1^d = \prtc{ 2 }, \set{H}_2^d = \prtc{ 3 }, \set{H}_3^d = \prtc{ 4 }$ for $d \in \set{D}$}\vspace{-4.5mm}
  \label{fig:IdealLineNet}
\end{figure}

A line network in Fig. \ref{fig:IdealLineNet} is simulated with the interval length $T = 100$ and the constant $K = 10$.   The network is simulated for $10^5$ slots.  The average backlogs shown in Table \ref{tlb:IdealLineNet} are calculated after Algorithm \ref{alg:Main} and MaxWeight converge.  The link-capacity sharing of Algorithm \ref{alg:Main} leads to small queue backlogs.  For scaling comparison, when $T = 1000$, the average backlogs of commodity 1 at switch 1 are respectively $23.94$ and $16005.63$ under Algorithm \ref{alg:Main} and MaxWeight.  In this simulation, the maximum value of $k_{ij}(t)$ for all $i \in \set{S}, j \in \set{H}_i, t \in \prtc{0, \dotsc, 10^5}$ is $1.732 < K$.

\begin{table}
  \caption{Average backlogs under Algorithm \ref{alg:Main} and MaxWeight}
  \label{tlb:IdealLineNet}
  \begin{tabular}{ c | r | r | r | r |}\cline{2-5}
    & \multicolumn{2}{ c | }{Commodity 1} & \multicolumn{2}{ c |}{Commodity 2} \\\cline{2-5}
    & Algorithm \ref{alg:Main} & MaxWeight & Algorithm \ref{alg:Main} & MaxWeight  \\\hline
    \multicolumn{1}{| c |}{Switch-1} & $14.88$ & $2598.19$ & $6.25$ & $213.45$ \\
    \multicolumn{1}{| c |}{Switch-2} &  $8.17$ & $1699.99$ & $2.35$ & $195.12$ \\
    \multicolumn{1}{| c |}{Switch-3} &  $7.52$ &  $700.01$ & $2.16$ & $199.51$ \\
    \multicolumn{1}{| c |}{Switch-4} &  $7.12$ &    $7.00$ & $2.00$ &   $2.00$ \\\hline
  \end{tabular}\vspace{-2.5mm}
\end{table}

A network in Fig. \ref{fig:IntraDC} is simulated with $T=100$ and $K=10$.  After $10^5$ slots, the average backlogs per queue (from all $124$ queues) are $34.23$ and $555.95$ under Algorithm \ref{alg:Main} and MaxWeight.  An event $k_{ij}(t) = K$ occurs $93.41\%$ of the times due to that the network operates near its capacity boundary.  Reducing the arrivals by $12\%$ ($24\%$) yields $20.68\%$ ($5.49\%$) of the times that $k_{ij}(t) = K$.  This suggests that Algorithm \ref{alg:PacketFilling} is rarely invoked, i.e., $k_{ij}(t) < K$, when a network does not operate near its capacity boundary.  In practice, TCP flows with congestion control are different from the i.i.d. arrivals, so Algorithm \ref{alg:PacketFilling} is not included in the heuristic algorithm.

\begin{figure}
  \centering
  \includegraphics[scale=1.0]{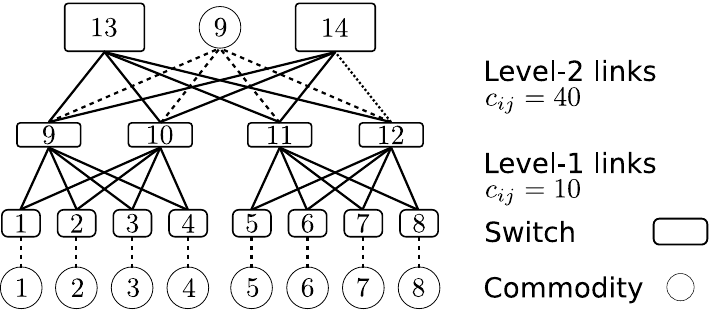}\vspace{-2mm}
  \caption{Intra datacenter network with $\set{S}=\prtc{1,2,\dotsc,14}, \set{D}=\prtc{1,2,\dotsc,9}$.  Each next-hop set $\set{H}_i^d$ contains next-hop switches with the shortest distance to commodity $d$, e.g., $\set{H}_1^8 = \prtc{9,10}, \set{H}_9^8 = \prtc{13,14}, \set{H}_{13}^8 = \prtc{11,12}, \set{H}_{11}^{8} = \prtc{8}, \set{H}_1^9 = \prtc{9,10} = \set{H}_1^e$ for $e \in \prtc{2,3,4}$.  Arrivals are $\expect{a_i^d(t)} = 2$ and $\expect{a_i^9(t)} = 1$ for $d,i \in \prtc{1,\dotsc,8}$; otherwise 0.  Departure rate is $b_i^d = 20$ if commodity $d$ connects to switch $i$; otherwise $0$.}\vspace{-4.5mm}
  \label{fig:IntraDC}
\end{figure}

\subsection{Network Simulator}
The heuristic in-network load-balancing algorithm in Section \ref{sec:Implementation} is simulated by OMNeT++ \cite{OMNet}.  All simulations share the following setting.  Capacity of each commodity queue at a switch port is $200$ packets.  For ECMP setting, a shared queue at a switch port has buffer capacity of $200\times\abs{\set{D}}$ packets, where $\set{D}$ is a set of commodities in a considered network.  Configuration interval is $T = 1$ms, and the scaling parameter is $\alpha = 5$.  The NewReno TCP from INET Framework \cite{INET} is adjusted for $10$Gbps and $40$Gbps link speeds.  Every TCP flow is randomly established during $[0s, 0.5s]$ and starts during $[1s, 1.01s]$.  Each flow transmits $1$MB of data.  Flow completion time (FCT) is measured as the performance metrics, which are also used in \cite{CONGA,DeTail}.  FCT is the duration of time to complete a flow, i.e., the time to send $1$MB of data.

A network without commodity 9 in Fig. \ref{fig:IntraDC} is simulated.  The speeds of level-1 links and level-2 links are respectively $10$Gbps and $40$Gbps.  $64$ flows are generated from a commodity to one another, and the total number of flows in the network is $3584$.  The FCTs under the heuristic algorithm and ECMP are shown in Fig. \ref{fig:IntraSymFCT}, and their variances are $5.8 \times 10^{-4}$ and $19.5 \times 10^{-4}$.  Unsurprisingly, they are comparable, since the topology is optimized for ECMP.  The FCTs under the heuristic algorithm has less variation, as the distribution of flows are more balance.  Note that the tail of FCTs is critical for interactive services \cite{DeTail}.

\begin{figure}
  \centering
  \includegraphics[scale=0.26]{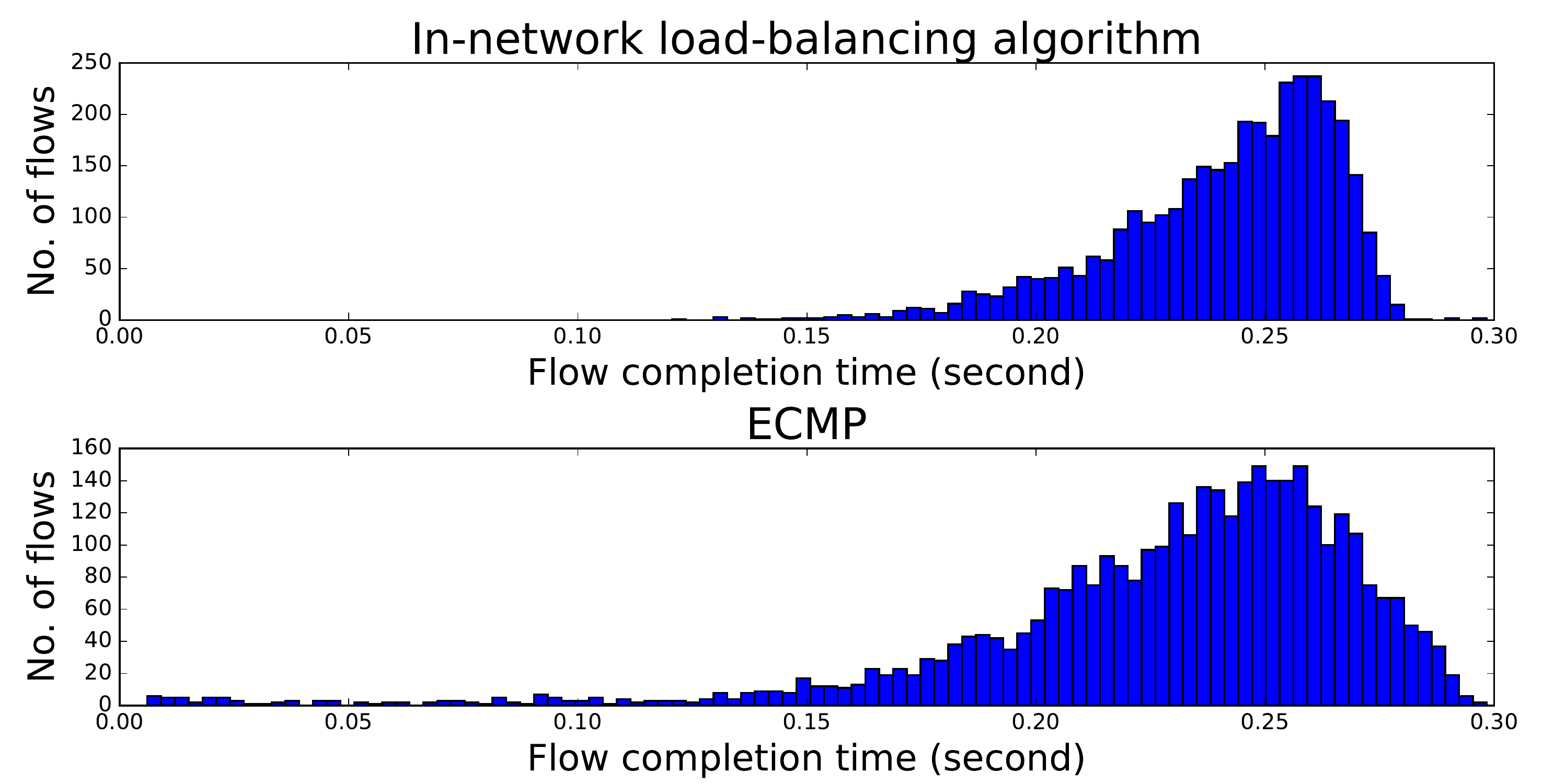}\vspace{-2mm}
  \caption{The FCTs from the network in Fig. \ref{fig:IntraDC} without commodity $9$}\vspace{-2mm}
  \label{fig:IntraSymFCT}
\end{figure}

A follow-up scenario is simulated when the link between switches $12$ and $14$ in Fig. \ref{fig:IntraDC} fails.  The FCTs of all flows are shown Fig. \ref{fig:IntraAsymFCT}.  The heuristic algorithm balances the flows better than ECMP.  Switches $9$ and $10$ hash more flows to switch $13$ than switch $14$, while ECMP hashes flows equally.

\begin{figure}
  \centering
  \includegraphics[scale=0.26]{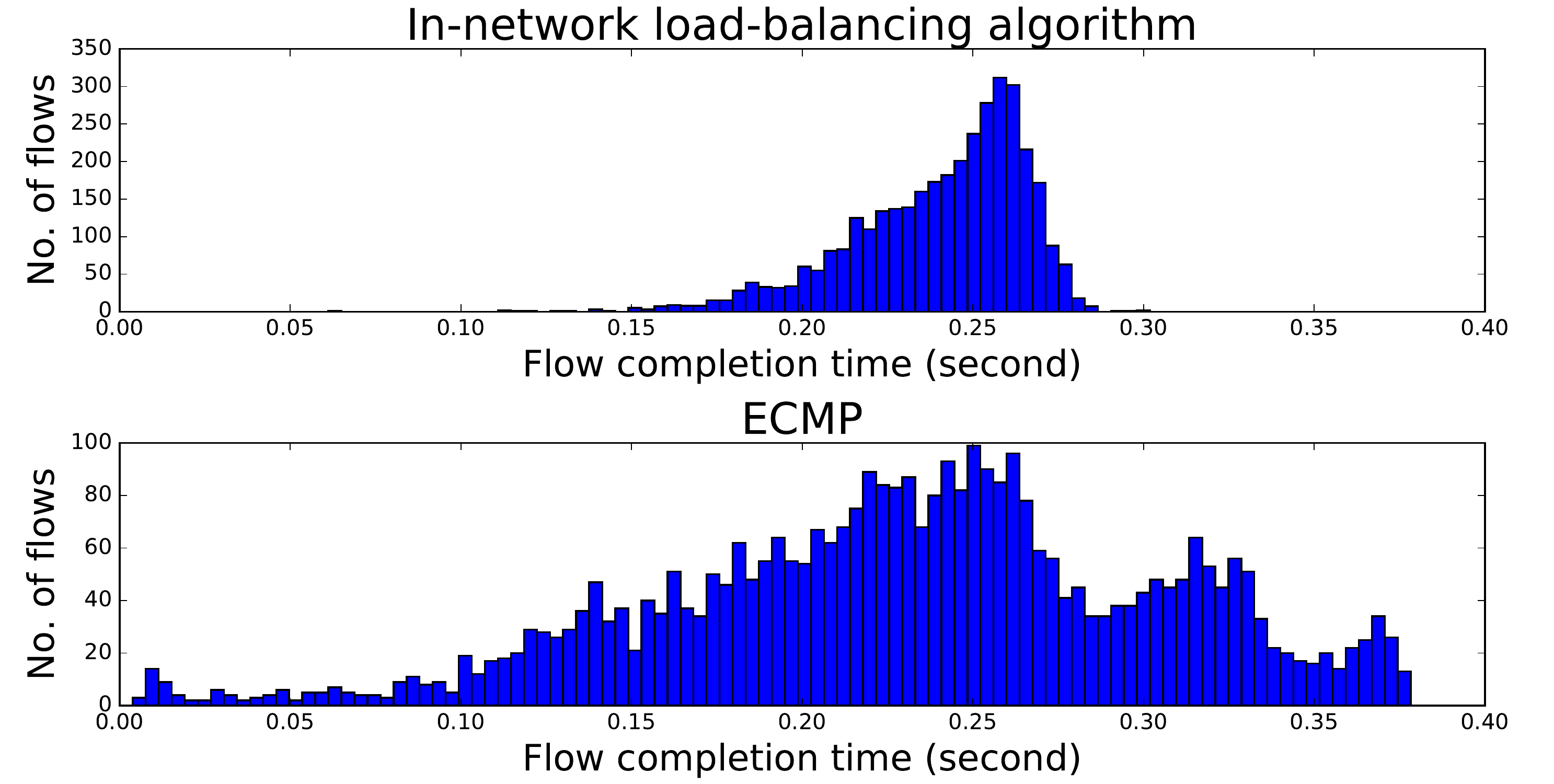}\vspace{-2mm}
  \caption{The FCTs of all flows in the network in Fig. \ref{fig:IntraDC} where commodity $9$ is omitted and the link between switches $12$ and $14$ fails }\vspace{-4.5mm}
  \label{fig:IntraAsymFCT}
\end{figure}

A network in Fig. \ref{fig:Priority} illustrates the adaptiveness of the heuristic algorithm when some link capacity is taken away by priority flows.  The FCTs of all flows are shown in Fig. \ref{fig:PriorityFCT}.  The FCTs under the heuristic algorithm is more balance compared to the FCTs under ECMP, as switches $1$ and $2$ hash more flows to switch $4$ instead of equally hashing in ECMP case.  Note that if the priority flows begin shortly after $1.01$s, the same result is observed.

\begin{figure}
  \centering
  \includegraphics[scale=1.0]{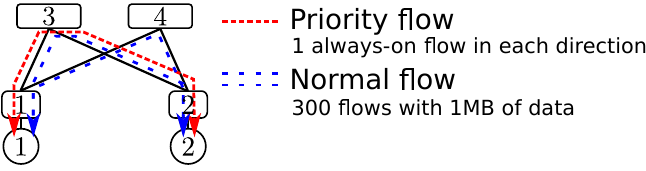}\vspace{-2mm}
  \caption{A network with a priority flow in each direction and $600$ normal flows between commodities $1$ and $2$}
  \label{fig:Priority}
\end{figure}

A similar trend can be observed from scenarios with short flows ($10$KB of data per flow) when $T = 0.1$ms and $\alpha = 1$.

\begin{figure}
  \centering
  \includegraphics[scale=0.26]{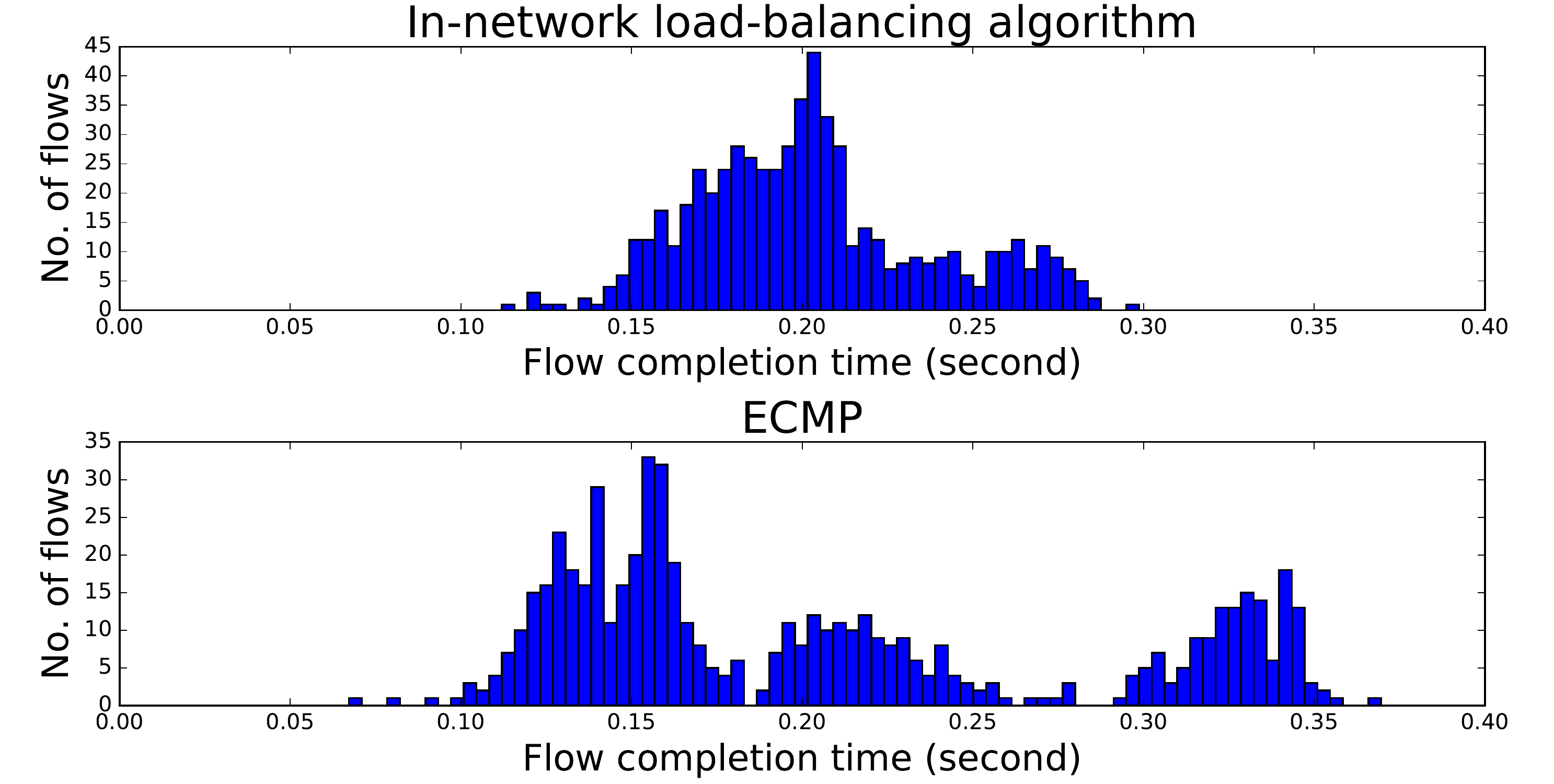}\vspace{-2mm}
  \caption{The FCTs of normal flows between commodities $1$ and $2$ in Fig. \ref{fig:Priority}}\vspace{-4.5mm}
  \label{fig:PriorityFCT}
\end{figure}

\section{Conclusion}
\label{sec:Conclusion}
This paper showed that practical load-balancing in datacenter networks can be tackled by the concept of throughput optimality.  
Our first algorithm is a novel variation of the MaxWeight concept that treats control plane and data plane timescales (useful for software defined networking), 
allows link capacity sharing during a control plane interval, incorporates weighted fair queueing aspects, 
and comes with a proof of throughput optimality.  Next, this algorithm was modified to include heuristic 
improvements that allow easy operation with practical switch capabilities and works gracefully with TCP flows.  
Ideal and OMNeT++ simulations show promising potential against existing MaxWeight and ECMP.

\bibliographystyle{IEEEtran}
\bibliography{reference}

\end{document}

%% file: SuchaLib.tex
\newcommand{\set}[1]{ {\mathcal{#1}} }

\newcommand{\expect}[1]{ {\mathbb{E}\left[{#1}\right]} }

\newcommand{\abs}[1]{ {\left|{#1}\right|} }

\newcommand{\norm}[1]{ {\left\lVert{#1}\right\rVert} }

\newcommand{\round}[1]{ {\Round\left[{{#1}}\right]} }

\newcommand{\IntSet}{ {\mathbb{Z}} }

\DeclareMathOperator{\argmax}{argmax}

\DeclareMathOperator{\Argmin}{Argmin}

\DeclareMathOperator{\Round}{round}

\newcommand{\maximize}{ {\text{Maximize}} }
\newcommand{\minimize}{ {\text{Minimize}} }
\newcommand{\subjectto}{ {\text{Subject to}} }

\newcommand{\prts}[1]{ {\left[{#1}\right]} }
\newcommand{\prtr}[1]{ {\left({#1}\right)} }
\newcommand{\prtc}[1]{ {\left\{{#1}\right\}} }

\newcommand{\maxvar}{ {(\text{max})} }
